\documentclass[letterpaper, 10 pt, conference]{ieeeconf}  
\usepackage{graphicx, amsfonts, amsmath, amssymb}%}
\IEEEoverridecommandlockouts % Comment this line out % Comment this line out
\usepackage{graphics}
\DeclareMathOperator*{\argmin}{argmin} 
\newtheorem{theorem}{Theorem}
\newtheorem{lemma}{Lemma}
\newtheorem{problem}{Problem}
\newtheorem{remark}{Remark}

\title{\LARGE \bf
 Privacy Constrained Information Processing
}
\author{Emrah Akyol, C\'edric Langbort, Tamer Ba\c{s}ar% <-this % stops a space
\thanks{Authors are with the Coordinated Science Laboratory, University of Illinois at Urbana-Champaign, 1308 West Main Street, Urbana, IL 61801, USA
        {\tt\small \{akyol, langbort, basar1\}@illinois.edu}}   
\thanks{This work was supported by AFOSR MURI Grant FA9550-10-1-0573.}
}

\begin{document}

\maketitle
\thispagestyle{empty}
\pagestyle{empty}

%%%%%%%%%%%%%%%%%%%%%%%%%%%%%%%%%%%%%%%%%%%%%%%%%%%%%%%%%%%%%%%%%%%%%%%%%%%%%%%%
\begin{abstract}
 This paper studies communication scenarios where the transmitter  and the receiver have different objectives due to privacy concerns, in the context of a variation of the strategic information transfer (SIT) model of Sobel and Crawford. We first formulate the problem as the minimization of a common distortion by the transmitter and the receiver subject to a privacy constrained transmitter. We show the equivalence of this formulation to a Stackelberg equilibrium of the SIT problem. Assuming an entropy based privacy measure, a quadratic distortion measure and jointly Gaussian variables, we characterize the Stackelberg equilibrium. Next, we consider asymptotically optimal compression at the transmitter which inherently provides some level of privacy, and study equilibrium conditions.  We finally analyze the impact of the presence of an average power constrained Gaussian communication channel between the transmitter and the receiver on the equilibrium conditions.  %We finally analyze the setting where the receiver has access to a side information, correlated with the privacy sequence. We particularly focus on the setting where the transmitter controls the  receiver side information about private sequence, rather surprisingly, show that transmitter may benefit from the presence of this side information, depending on the problem parameters.  
% This paper considers the problem of privacy in communication and control systems with partial privacy concerns-where only a subset of agents are concerned with privacy. We first formulate  the privacy constrained communication problem, and next characterize the St
% 
%  equivalency of this formulation and the strategic information transfer (SIT) problem of Crawford and Sobel in economics. 
% 
% 
% Several results follow directly from this equivalence, including the optimality of deterministic encoding policies in the deterministic private information settings. Focusing on the random private information, we characterize the achievable distortion-privacy region. While  the computation of this region is difficult in general, we for the important case of the jointly Gaussian source-private information and quadratic encoding/decoding distortion measures. As a byproduct, we show sub-optimality of the quantizer based encoding strategies for this setting and prove the optimality of Gaussian test channel-like encoding-paralleling the well-known rate-distortion optimal encoding of a Gaussian source under quadratic distortion. Finally, we show the that optimality of the zero-delay encoding for the information theoretic setting with quadratic measures. % Finally, we extend our analysis to vector spaces in two different common objective and privacy constraint models. 

  \end{abstract}

%%%%%%%%%%%%%%%%%%%%%%%%%%%%%%%%%%%%%%%%%%%%%%%%%%%%%%%%%%%%%%%%%%%%%%%%%%%%%%%%
\section{Introduction}
This paper studies communication scenarios where the transmitter  and the receiver  have different objectives due to privacy concerns of the transmitter. Consider, for example,   the communication between a transmitter and a receiver, where the common objective of both agents is to minimize some objective function. However, the transmitter has an additional objective: Convey as little (accurate) information as possible about some privacy related information- correlated with the transmitted message- since the reconstruction at the receiver is reported into databases visible to other parties (government agencies, police, etc.). Obviously, the receiver is {\it oblivious} to this objective, i.e., privacy is not a common goal. Then, what kind of transmitter and receiver mappings (encoders and decoders) yield equilibrium conditions? How do compression at the transmitter or the presence of a noisy channel impact such equilibria? %How does the receiver side information (not available at the encoder) effect equilibrium conditions?  

Such problems where better informed transmitter communicates with a receiver who makes the ultimate decision concerning both agents have been considered in the economics literature under the name of  ``cheap talk" or strategic information transfer (SIT), see e.g.,  \cite{crawford1982strategic, battaglini2002multiple} and the references therein. The SIT problem\cite{crawford1982strategic}, involves settings where the private information, available only to the transmitter, affects the  transmitter utility function. The receiver utility does not depend on this private information and thus is different from that of  the transmitter. The objective of the agents, the transmitter and the receiver, is to maximize their respective utility functions.  One of the main results of  \cite{crawford1982strategic} is that,  all Nash equilibrium  points can be achieved by a quantizer as a transmitter strategy.  Here, motivated by the conventional communication systems design, we analyze the Stackelberg equilibrium \cite{BasarBook}, where the receiver knows the encoding mappings and optimizes its decoding function accordingly. This fundamental difference between the two problem settings enables in the current case the use of Shannon theoretic arguments to study the fundamental limits of compression and communication in such {\it strategic} settings. In \cite{farokhi2014gaussian}, the set of  Stackelberg equilibria was studied for estimation with biased sensors. Here, we study the communication and compression with privacy constraints in the same context. %, i.e., the transmitter is the leader and the 
%This result essentially connects such game theoretic problems to source coding (with mismatched encoder and decoder distortion measures) where quantization is the main mathematical tool. 
\begin{figure*}
\centering
\includegraphics[scale=0.4]{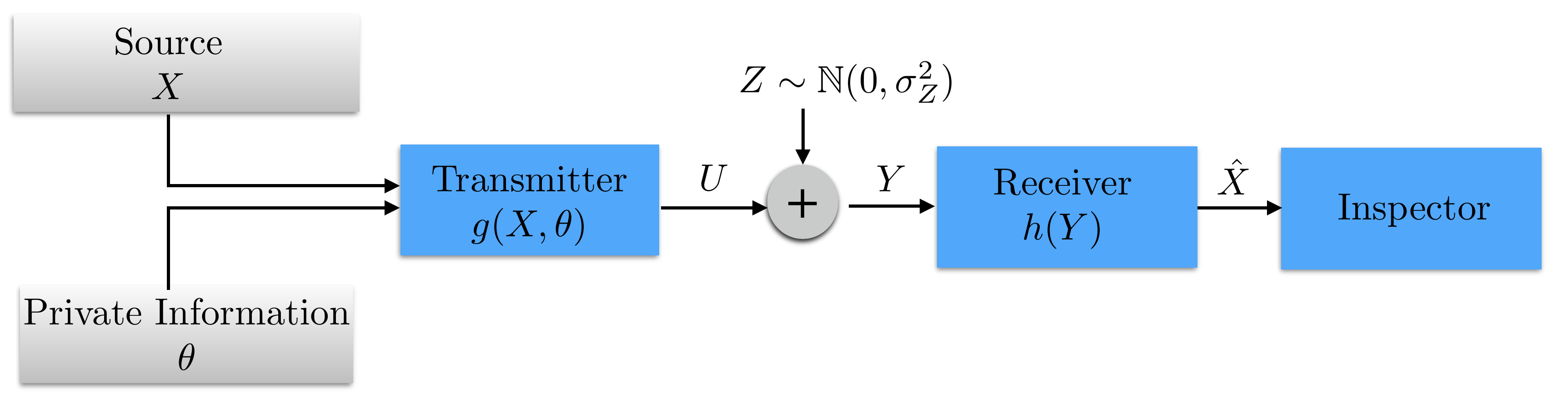}
\caption{The problem setting}
\label{fig:1}
\end{figure*}

The privacy considerations have recently gained renewed interest, see e.g., \cite{dwork2011differential,mcsherry2007mechanism,lalitha} and the references therein. %Note that the problem considered here is fundamentally different than the secrecy class of problems in information theory (e.g., Wyner's wiretap channel etc, \cite{})  since we assume that there are no eavesdroppers  but only an intended receiver. 
In  \cite{yamamoto1988rate,yamamoto2}, Yamamoto studied a compression problem similar to the one considered here: find an encoder  such that there exists a decoder  that guarantees a distortion no larger than $D_C$ when measured with $\rho_C$ and at the same time cannot be smaller than $D_P$, when measured with $\rho_P$ in conjunction with {\it any other} decoder.  In \cite{lalitha}, Yamamoto's result was extended to some special cases to analyze the privacy-utility tradeoff in databases. 

In this paper, we explicitly study the equilibrium conditions under  transmitter's privacy constraints. The contributions of this paper are: 
\begin{itemize}
\item We first formulate the problem, that involves minimization of a global objective by the encoder and the decoder subject to a privacy constraint  measured by a different function. 
\item Assuming an entropy based privacy measure and quadratic distortion measure, we characterize the achievable distortion-privacy region with or without compression at the transmitter. 
\item We  study the impact of the presence of an average power constrained Gaussian communication channel on the privacy-distortion trade-off. 
\end{itemize}

% 

% 
%This paper is organized as follows. In Section II, we present the problem definition and preliminaries, and we also review the prior information theoretic results related to source coding with mismatched measures. In Section III, we present an alternative proof of optimality of deterministic transmission policies in the SIT problem. In Sections IV, V, VI we respectively present an alternative problem definition, an information theoretic approach and a quantizer design approach to the SIT problem.   Finally, we discuss the future directions in Section VII. 

\section{Preliminaries}
\subsection{Notation}
Let $\mathbb R$ and $\mathbb R^+$ denote the respective sets of real numbers and positive real numbers. Let  $\mathbb E(\cdot)$ denote the expectation operator.  The Gaussian density with mean $ \mu$ and variance  $\sigma^2$ is denoted as $\mathcal N( \mu,\sigma^2)$.  All logarithms in the paper are natural logarithms and may in general be complex valued, and the integrals are, in general, Lebesgue integrals.  Let us define  ${\cal {S}}$ to denote the set of Borel measurable, square integrable  functions $\{ { f}: \mathbb R \rightarrow \mathbb R$\}. For information theoretic quantities, we use standard notations as, for example,  in \cite{CoverBook}. 
%Let $\{X_t\}_{t=1}^\infty$, $X_t \in {\cal X}$, be a discrete memoryless source (DMS) with generic distribution $P_X(x)$.
% The vector $[X(1), X(2),..., X(n)]$ is compactly denoted by $X^n$.
%Let ${\cal Y}$ denote the reproduction alphabet.  
We let $H(X)$ and $I(X;Y)$  denote the entropy of a discrete random variable $X$(or differential entropy if $X$ is continuous), and the mutual information between the random variables $X$ and $Y$, respectively.

\subsection{Setting-1: Simple Equilibrium}
We consider the general communication system whose block diagram is shown in Figure 1. The source $X$ and private information $\theta$ are mapped into ${Y}\in  \mathbb R$   which is fully determined by the conditional distribution $p(\cdot| x,\theta)$. For the sake of brevity, and with a slight abuse of notation, we refer to this as a stochastic mapping  $Y= g(X, \theta)$ so that 
\begin{equation}
\mathbb P( g(X, \theta)\in {\cal Y})=\int \limits_{ y'\in {\cal Y}} p( y'| x, \theta)\mathrm d x \mathrm d \theta\,\,\quad  \forall {\cal Y} \subseteq \mathbb R
\end{equation}
 holds almost everywhere in $X$ and $\theta$. 
 %by a function $ { g_T (\cdot)}\in {\cal {S}}_{m}^{k} $ and transmitted over an additive noise channel. 
Let the set of all such mappings be denoted by $\Gamma$ (which has a one-to-one correspondence to the set of all the conditional distributions that construct the transmitter output $Y$). 

The receiver produces an estimate of the source $ {\hat X}$ through a mapping $ h \in {\cal S}$ as $ {\hat X}= h(Y)$. An inspector observes the estimate of the receiver, aims to learn about the private information $\theta$, i.e., minimize $H(\theta|\hat X)$.  Note that the joint statistics of the random variables is common knowledge. The common objective of the transmitter and the receiver is 
to minimize end-to-end distortion measured by a given distortion measure $\rho_C$ as
\begin{equation}
D_C=\mathbb E\{\rho_C( X, {\hat X})\}
\end{equation}
over the mappings $ g, h$  subject to a privacy constraint: 
\begin{equation}\label{constraint}
\mathbb E\{\rho_P( \theta, {\hat X})\}\geq J_P
\end{equation}
{\it over only the encoding mapping $ g$} (the decoder is oblivious to the privacy objective). Here, the encoder aims to minimize $D_C$  in collaboration with the decoder (classical communication  problem). The encoder has another objective, however: to maximize privacy,  measured by, say $\rho_{P}$ or  to guarantee that this privacy is  not less than a given threshold, say $J_P \in \mathbb R^+$. Note that the decoder has no interest in finding out this information, or in satisfying or not satisfying this constraint.   
 This subtle difference, i.e., the fact that there is a mismatch between the objectives of the decoder and those of the encoder, motivates us to consider  this problem in a game theoretic setting (the SIT problem). 
 In game theoretic terms, we consider a {\it constrained Stackelberg game} where only one of the players (the encoder) is concerned with the constraint in (\ref{constraint}). Here, the encoder knows that the decoder will act to minimize the global cost $D_C$. 
 Hence, Player 1 (leader) is the encoder and it knows that Player 2 (follower, the decoder) acts to minimize $D_C$. 
  The leader (the encoder) acts to minimize $D_C$ subject to (\ref{constraint}) {\it knowing}  the decoder's objective. 
  In the following, we present this optimization problem formally:

\begin{problem}
Find $ g(\cdot,\cdot) \in \Gamma$ which minimizes 
$$\mathbb E\{\rho_C( X,  h^*( g( X,\theta)) )\}$$
subject to 
$$\mathbb E\{\rho_P( \theta,  h^*( g( X,\theta)) )\} \geq J_P$$
where 
$$ h^*(g)=\argmin \limits_{\substack h \in {\cal S}}\mathbb E\{\rho_C( X,  h( g( X, \theta))\}$$
\end{problem}

In this paper, we specialize to quadratic Gaussian settings, i.e., the source and the private information are jointly Gaussian, the distortion measure is  mean squared error and the privacy measure is conditional entropy. Particularly, this setting implies that  $(X,\theta) \sim \mathcal N( 0, R_{X\theta})$, where $ R_{X\theta}=\sigma_X^2 \left[ \begin{array}{cc}  1 & \rho \\ \rho & r \end{array} \right ]$; without loss of any generality we take $\rho \geq 0$, and naturally also $\rho^2\leq r$.  Further, the distortion and privacy measures are given as follows: 
\begin{align}
\rho_C(x,y)=(x-y)^2 
\end{align}
and 
\begin{align}
\rho_P(x,y)=-\log p(X|Y)
\end{align}
which results in conditional entropy as the privacy measure  $\mathbb E\{\rho_P(x,y)\}=H(X|Y)$. 
The following lemma is a simple consequence of the fact that the Gaussian distribution maximizes entropy subject to covariance constraints (see \cite{diggavi}). This lemma will be used to convert the equilibrium conditions related to compression and communication problems to a control theoretic framework (an optimization problem involving  second order statistics).
\begin{lemma}
At equilibrium, $Y, X, \theta$ are jointly Gaussian. 
\label{lemma:1}
\end{lemma}
Lemma \ref{lemma:1} ensures optimality of a linear decoder and hence $H(\theta|\hat X)=H(\theta|Y)$ since $\hat X$ is an invertible function of $Y$. Note that invertibility of the decoding mapping, and hence this simplification in the privacy constraint is a direct consequence of the Stackelberg equilibrium. The Nash equilibrium variant of the same problem, studied in \cite{crawford1982strategic} without a privacy constraint,  does not yield $H(\theta|\hat X)=H(\theta|Y)$, since the decoding mapping at equilibrium is not invertible (quantizer based).  Lemma \ref{lemma:1}  and the fact that maximizing $H(\xi_1|\xi_2)$ is equivalent to maximizing $\mathbb E\{( \xi_1-\mathbb E \{\xi_1|\xi_2\} )^2\}$ for jointly Gaussian $\xi_1, \xi_2$, enable the following reformulation of Problem 1: 
\begin{problem}
Find $Y= g(X, \theta)$ where $g(\cdot, \cdot) \in \Gamma $  minimizes 
$$D_C=\mathbb E\{( X-\mathbb E \{X|Y\} )^2\}$$
subject to 
$$\mathbb E\{( \theta-\mathbb E \{\theta|Y\} )^2\} \geq D_P$$
\end{problem}
%Problem 2 has an intuitive interpretation: The transmitter 
In this paper,  we show the existence of such an equilibrium, and its essential uniqueness
\footnote{The optimal transmitter and receiver mappings are not strictly unique, in the sense that multiple trivially ``equivalent" mappings can be used to obtain the same MSE and privacy costs. For example, the transmitter can apply any invertible mapping $\gamma(\cdot)$ to $Y$ and the receiver applies $\gamma^{-1} (\gamma (Y))=Y$  the prior to $h(\cdot)$. To account for such trivial, essentially identical solutions, we use the term ``essentially unique".}.

%provided that such a saddle point exists. I
%
\subsection{Setting-2: Compression}
Next, we consider the compression of the source $X$ subject to privacy constraints,  and analyze this problem from an information theoretic perspective. 
%Note that our primary objective here is to use well-known information theoretic results to analyze control theoretic problem settings involving privacy. 
 
%
%We first note that all objectives are convex in encoding/decoding mappings, hence there is no duality gap.  This implies that  there exists a $T\in \mathbb R^+$  such that minimizing $\mathbb E\{\rho_C(X, h^*(g(X))\}$  subject to $\mathbb E\{\rho_P(X, h^*(g(X))\} \geq D_P$ ($D_P$ is determined by  $T$) is equivalent to minimizing the Lagrangian\cite{rockafellar1997convex}: 
Formally, we consider an i.i.d. source $X^n$ and a private information sequence $\theta^n$ to be  compressed to $M \approx2^{nR}$ indices through $f_E$. The receiver applies a decoding function $f_D$ to generate the reconstruction sequence ${\hat X}^n$. Due to the {\it strategic} aspect of the problem, we have one distortion measure and one privacy measure. Similar to previous settings, we assume that the distortion is measured by MSE 
\begin{equation}\label{decoder}
D_C(f_E, f_D)=\frac{1}{n} \sum \limits_{i=1}^{n}\mathbb E \{(X_i-{\hat X}_i)^2\},
\end{equation}
and privacy is measured by conditional entropy.
\begin{equation}\label{encoder}
J_P(f_E, f_D)=H(\theta^n|{\hat X^n})
\end{equation}  

 A triple $(R, D_C, J_P)$ is called {\em  achievable} if for every $\delta>0$ and sufficiently large $n$, there exists a block code $(f_E, f_D)$ such that
\begin{eqnarray*}
\frac{1}{n}\log { M} & \leq & R + \delta \\
D_C(f_E, f_D) & \leq & D_C + \delta \\
J_P(f_E, f_D)& \leq & J_P + \delta \; .
\end{eqnarray*}
The set of achievable rate distortion triple $(R, D_C, J_P)$ is denoted here as ${\cal RD}$. The following theorem, whose proof directly follows from the arguments in \cite{yamamoto1988rate} for a general distortion measure $d_C$ and conditional entropy,  characterizes the achievable region ${\cal RD}_{S}$, i.e., converts the problem from an $n-$dimensional optimization to a single letter.  
 \begin{theorem}
\label{sc}
${\cal RD}$  is the convex hull of triples $(R, D_C, J_P)$ for
  \begin{align}
  R&\geq I(X,\theta ; Y) \nonumber 
  \end{align}
 for a conditional distribution $p(Y|X,\theta)$ and a deterministic decoding function $h:\mathbb R\rightarrow \mathbb R$ which satisfy
  \begin{align}
 D_C&\geq \mathbb E\{d_C(X, h(Y))\} \nonumber\\
J_P&\leq H(\theta|Y) \nonumber
  \end{align}
  \end{theorem}
%In this paper, we focus on jointly Gaussian sources and MSE distortion. Plugging these and noting that Gaussian distribution (similar arguments also appear in \cite{lalitha}) maximizes entropy subject to covariance constraints, and replacing conditional entropy with the equivalent second order constraint, we have the following characterization of ${\cal RD}$: 
%
% \begin{theorem}
%${\cal RD}$  is the convex hull of triples $(R, D_C, D_P)$ for
%  \begin{align}
%  R&\geq I(X,\theta ; Y) \nonumber 
%  \end{align}
% for a conditional distribution $p(Y|X,\theta)$ and a deterministic decoding function $h:\mathbb R\rightarrow \mathbb R$ which satisfy
%  \begin{align}
% D_C&\geq \mathbb E\{d_C(X, h(Y))\} \nonumber\\
%D_P&\leq H(\theta|Y) \nonumber
%  \end{align}
%  \end{theorem}
\subsection{Setting-3: Communication over Noisy Channel}
Finally, we consider an additive Gaussian noise between the transmitter and the receiver. This problem setting is shown in Figure 1, where the receiver observes $Y=U+Z$, where $Z \sim \mathbb N(0, \sigma_Z^2)$ is zero-mean Gaussian and distributed independent  of $X$ and $\theta$. Again, we focus on entropy based privacy and quadratic distortion measure and Gaussian variables.  The problem can then be reformulated as:

\begin{problem}
Find $U= g(X, \theta)$   that minimize
$$D_C=\mathbb E\{( X-\mathbb E \{X|Y\} )^2\}$$
subject to 
$$\mathbb E\{( \theta-\mathbb E \{\theta|Y\} )^2\} \geq D_P$$
where $Y=U+Z$.
\end{problem}
%As a special case, we also consider the case where there is no broadcasting side channel, i.e., effectively $\sigma_N^2\rightarrow \infty$.

\section{Main Results}

\subsection{Simple Equilibrium}
Note that Lemma 1 does not provide the exact form of the function $g(\cdot,\cdot) \in \Gamma$, although it implies that $Y=X+\alpha \theta+S$ for some $\alpha \in \mathbb R$ and  $S\sim \mathbb N(0,\sigma_S^2)$ independent of $X$ and $\theta$.  The following observation involves the two extreme cases of this problem, i.e., the endpoints of the $D_C$-$D_P$ curve.   
\begin{lemma}
At maximum privacy, where $H(\theta|\hat X)=H(\theta)$, and at minimum privacy, where $H(\theta|\hat X)=H(\theta|X)$, the equilibrium is achieved at   $Y=X+\alpha \theta$  for some $\alpha\in \mathbb R$. In other words, at end points, there is no need to have the  noise term $S$. 
\label{lemma:2}
\end{lemma}
\begin{proof}
At minimum privacy, obviously the optimal transmitter strategy is $Y=X$ which results in $D_C=0$. The only way to achieve maximum privacy, $H(\theta|\hat X)=H(\theta)$, is to render transmitter output $Y$ independent of $\theta$. Since variables are jointly Gaussian, this can be achieved by simply transmitting the prediction error, $Y=X+\alpha \theta$ where $\alpha=-\frac{\rho}{r} $ is MMSE prediction coefficient of $X$ from $\theta$. After prediction, the privacy constraint is satisfied  and adding noise only increases $D_C$, hence $Y=X+\alpha \theta$ is the optimal transmitter strategy. 
\end{proof}
In the following, we obtain auxiliary functional properties of $D_C$ and $D_P$ as a function of the encoding mapping $g(\cdot, \cdot)$ or equivalently $f_{Y|X, \theta}$ and the $D_C-D_P$ curve.

\begin{lemma}
$D_C$ and $D_P$ are concave functions of $f_{Y|X, \theta}$, and  $D_C(D_P)$ is an increasing, concave function of $D_P$. 
\end{lemma}

\begin{proof}
Let  $Y^{(i)}$ be the random variables achieving $D_C(f^{(i)}_{Y|X,\theta})$ be characterized by $f_{Y|X, \theta}^{(i)}$, and $h^{(i)}$ for $i=1,2$.

For  $0\leq c\leq 1$ we define 
$$ f^{c}_{Y|X,\theta}=c  f^{(1)}_{Y|X,\theta}+(1-c)  f^{(2)}_{Y|X,\theta}.$$ Then, for any decoding function $h^{(c)}(y)$
\begin{align}
   &c  D_C(f^{(1)}_{Y|X,\theta})+(1-c) D_C(f^{(2)}_{Y|X,\theta})   \nonumber \\
   &= c \int f_{X, \theta}(x,\theta) f^{(1)}_{Y|{X,\theta}}(x,\theta)  (x-h^{(1)}(y))^2 \mathrm d x \mathrm d \theta \nonumber \\ &+ (1-c) \int f_{X,\theta}(x,\theta)    f^{(2)}_{Y|X,\theta}(x,\theta)   (x-h^{(2)}(y))^2 \mathrm d x \mathrm d \theta\nonumber \\
   & \leq c \int f_{X,\theta}(x,\theta)   f^{(1)}_{Y|X}(x,y) (x-h^{(c)}(y))^2 \mathrm d x  \mathrm d \theta\nonumber \\ &+ (1-c) \int f_{X,\theta}(x,\theta)   f^{(2)}_{Y|X,\theta} (x,\theta) (x-h^{(c)}(y))^2 \mathrm d x\mathrm d \theta \nonumber \\ 
   & =  \int f_{X, \theta}(x,\theta)  ( c  f^{(1)}_{Y|X,\theta}(x,\theta)  +(1-c) f^{(2)}_{Y|X,\theta}(x,\theta)  ) \nonumber \\
      &\quad \quad\quad\quad\quad\quad\quad\quad\quad\quad\quad (x-h^{(c)}(y))^2 \mathrm d x\mathrm d \theta \nonumber \\
   & =  \int f_{X, \theta}(x,\theta)  f^{c}_{Y|X,\theta}(x,\theta)(x-h^{(c)}(y))^2 \mathrm d x\mathrm d \theta \nonumber \\
   & = D_C(c  f^{(1)}_{Y|X,\theta}+(1-c)  f^{(2)}_{Y|X,\theta})
\end{align}
which shows the concavity of $D_C$ in $f_{Y|X,\theta}$. Following similar steps, we obtain concavity of $D_P$ in $f_{Y|X,\theta}$, i.e.,  we have 
\begin{equation}
c  D_P(f^{(1)}_{Y|X,\theta})+(1-c) D_P(f^{(2)}_{Y|X,\theta}) \leq D_P( f^{c}_{Y|X,\theta})
\end{equation}

Note that  $D_C(D_P)$ is  non-decreasing  since $D_C(D_P)$ as expressed in Problem 2, is a minimization over a constraint set, as $D_P$ increases, minimization is performed over a smaller set, hence  $D_C(D_P)$ is non-decreasing. 

Toward showing concavity of $D_C(D_P)$, we first note that one can show concavity of $D_C$ in $f_{\hat \theta,Y|X,\theta}$ where $\hat \theta$  is the inspector's estimate of $\theta$, following similar steps to the preceding analysis. Then, let $Y^{(i)},{\hat \theta}^{(i)}$ be the random variables achieving $D_C(D_P^{(i)})$ be characterized by $f_{{\hat \theta},Y|X, \theta}^{(i)}$ for $i=1,2$. We need to show 
\begin{equation}
D_C(c D_P^{(1)}+(1-c) D_P^{(2)}) \geq c  D_C(D_P^{(1)})+(1-c) D_C(D_P^{(2)})
\end{equation}
for all $0\leq c\leq 1$. 
\begin{align}
 &c  D_C(D_P^{(1)})+(1-c) D_C(D_P^{(2)})= \nonumber\\ &c  D_C(f^{(1)}_{{\hat \theta},Y|X,\theta})+(1-c) D_C(f^{(2)}_{{\hat \theta},Y|X,\theta})    \\
   & \leq D_C(c  f^{(1)}_{{\hat \theta},Y|X,\theta}+(1-c)  f^{(2)}_{{\hat \theta},Y|X,\theta}) \\
   &= D_C(c D_P^{(1)}+(1-c) D_P^{(2)})
\end{align}
where the last step is due to the fact that $D_P$ is  linear in $f_{\hat \theta,Y|X,\theta}$. 

 %We follow the same steps as above and note that 
%\begin{equation}
% \int f_{X, \theta}(x,\theta)  f^{c}_{Y|X,\theta}(x,\theta)(x-h^{(c)}(y))^2 \mathrm d xd \theta \leq D_C(c D_P^{(1)}+(1-c) D_P^{(2)}
% \end{equation}
% due to concavity of  $D_P$ in $f_{Y|X,\theta}$. 
  Since  both privacy and distortions measures are continuous, we only needed to show concavity, since monotonicity is a consequence of concavity and continuity.
 \end{proof}

Lemma \ref{lemma:2} describes the equilibrium conditions at the end points.   The following theorem provides the exact characterization of this equilibrium over the entire $D_P-D_C$ region. 

 \begin{theorem}\label{main}
For the quadratic Gaussian setting with entropy base privacy constraint, the (essentially) unique equilibrium is achieved by $g(X,\theta)=X+\alpha \theta$ and $h(Y)=\kappa Y$ where $\alpha$ and $\kappa$ are constants given as: 
\begin{align} \alpha=&-\frac{\rho}{r} \pm \sqrt{\left (1-\frac{\rho^2}{r}\right)\left(\frac{1}{r}- \frac{\sigma_X^2}{D_P}\right)} \label{eq:15} \\
 \kappa=&\frac{1+\alpha\rho}{1+\alpha^2r+2\alpha\rho}
 \end{align}
 \end{theorem}
\begin{remark}
An interesting aspect of  the solution is that adding independent noise is {\it strictly} suboptimal in achieving the privacy-distortion trade-off. 
\end{remark}

\begin{proof}
First, we note that the optimal decoder mapping is 
\begin{equation}
h(Y)=\mathbb E\{X|Y\}
\end{equation}
regardless of the choice of encoder's policy $g(\cdot, \cdot)$. Hence, the problem simplifies to an optimization over the encoding mapping $g$. 

Noting that at equilibrium $Y$ and $X, \theta$ are jointly Gaussian, without loss of generality, we take  $Y=X+\alpha \theta+S$ where $S \sim \mathbb N(0, \sigma_S^2)$ is independent of $X$ and $\theta$. In the following, we find the value of $\alpha$ and $\sigma_S^2$ at equilibrium. First, let us express  $D_C$ and $D_P$ using standard estimation techniques:
\begin{align}
D_C(\alpha, \sigma_S^2)=&\mathbb E\{(X-\mathbb E\{X|Y\})^2\} \\
       =&\sigma_X^2 \left ( 1-\frac{(1+\alpha \rho)^2}{1+2\alpha \rho+\alpha^2 r+\frac{\sigma_S^2}{\sigma_X^2} }\right)
\end{align}

\begin{align}
D_P(\alpha, \sigma_S^2)=&\mathbb E\{(\theta-\mathbb E\{\theta |Y\})^2\} \\
      =&\sigma_X^2 \left ( r-\frac{(\rho+r\alpha )^2}{1+2\alpha \rho+\alpha^2 r+\frac{\sigma_S^2}{\sigma_X^2} }\right) \label{eq35}
\end{align}

\begin{figure}
\centering
\includegraphics[scale=0.4]{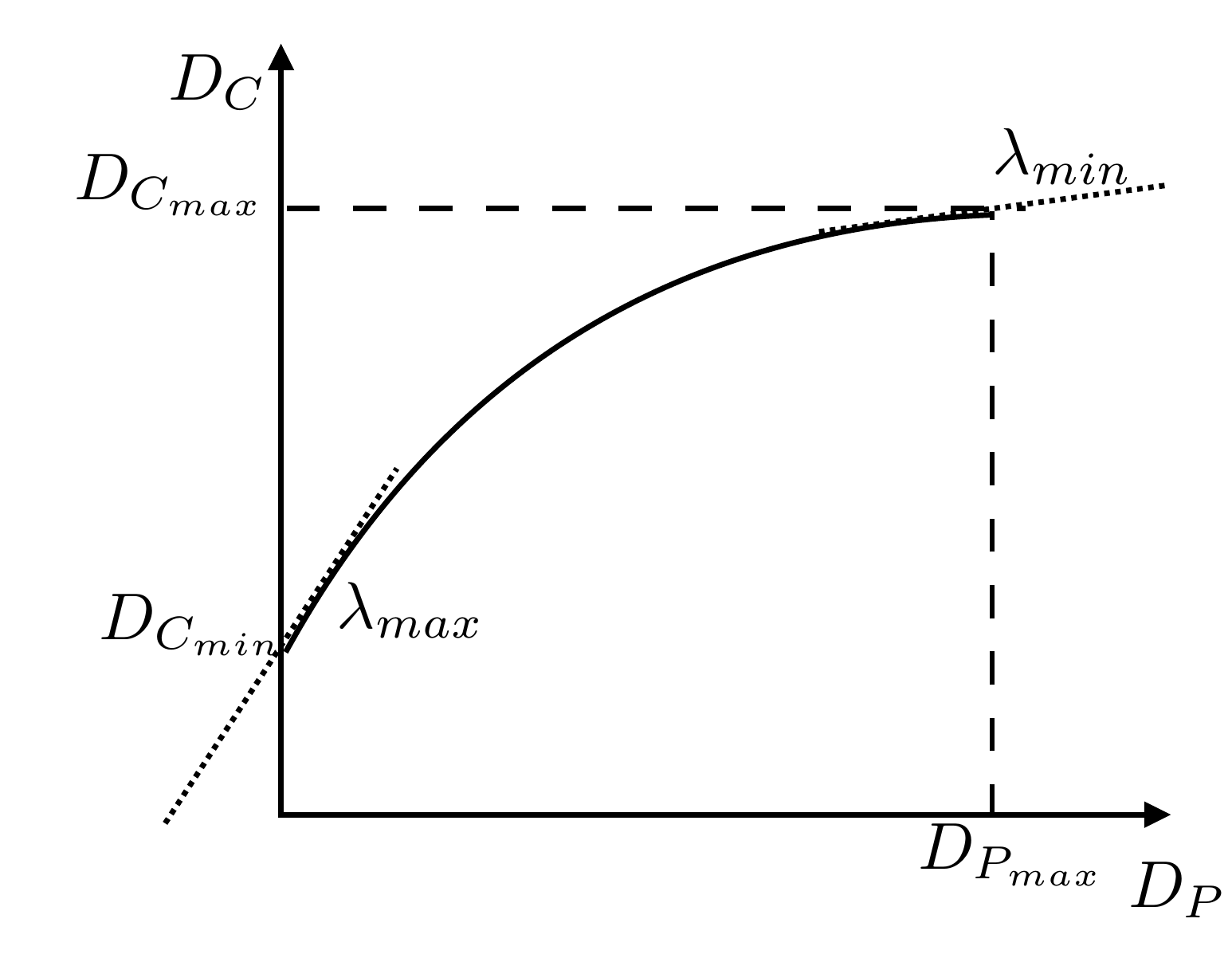}
\caption{$D_P-D_C$ curve}
\label{fig:2}
\end{figure}

Problem 1 can now be converted to an unconstrained \cite{BoydBook} minimization of the Lagrangian cost: 
\begin{equation}
 J(\alpha, \sigma_S^2)=D_C(\alpha, \sigma_S^2)-\lambda D_P(\alpha, \sigma_S^2)
\end{equation}
for where varying $\lambda \in \mathbb R^+$ provides solutions at different levels of privacy constraint $D_P$.
A set of necessary conditions for optimality can be obtained by applying K.K.T. conditions, one of which is  that $\lambda$ is the slope of the  curve:  
\begin{equation}
\lambda= \frac{d D_C(\alpha^*, \sigma_S^{2*})}{d D_P(\alpha^*, \sigma_S^{2*}) }
\end{equation}
at the optimal values of $\alpha$ and $\sigma_S^2$. Let us expand $J(\alpha, \sigma_S^2)$
\begin{equation}
J(\alpha, \sigma_S^2)= \sigma_X^2 (1-\lambda r)-\sigma_X^2 \left ( \frac{(1+\alpha \rho)^2-\lambda (\rho+r\alpha )^2}{1+2\alpha \rho+\alpha^2 r+\frac{\sigma_S^2}{\sigma_X^2} }\right) \nonumber
\end{equation}
The value of $\sigma_S^{2*}$ depends on the sign of $\Xi \triangleq \lambda (\rho+r\alpha )^2-(1+\alpha \rho)^2$, i.e., if $\Xi\geq0$, then $\sigma_S^{2*}=0$. In the following, we show that $\Xi\geq0$ for all  $\alpha  \in  [-\frac{\rho}{r},0]$, and hence $\sigma_S^{2*}=0$.

First, we note that the minimum value of $\lambda$ is reached at the maximum allowed $D_P$, as depicted in Figure 2. From Lemma \ref{lemma:2} for $D_P=D_{P_{max}}$, the solution implies that  $\sigma_S^{2*}=0$. Hence, $\Xi\geq0$ at $D_P=D_{P_{max}}$. Plugging the values, we obtain $\lambda \geq 0$. Following similar steps for $D_P=D_{P_{min}}$, we obtain $\lambda\leq \frac{1}{\rho^2}$. For  $\lambda \in [0,\frac{1}{\rho^2}]$ and $\alpha \in  [-\frac{\rho}{r},0]$, we have 
\begin{align}
\Xi \triangleq &(1+\alpha \rho)^2-\lambda (\rho+\alpha)^2 \nonumber\\
               \geq &(1+\alpha \rho)^2-\frac{1}{\rho^2} (\rho+\alpha)^2 \nonumber\\
                = &2 \alpha (\rho-\frac{r}{\rho})+\alpha^2(\rho^2-\frac{r^2}{\rho^2}) \nonumber\\
                = & \alpha (\rho-\frac{r}{\rho}) (2+\alpha(\rho+\frac{r}{\rho})).
\end{align}
Note that from Cauchy-Schwarz inequality, we have $r\geq \rho^2$ and noting that $-\frac{\rho}{r} \geq \alpha \geq 0$, we conclude $\Xi \geq 0$. Hence, $\sigma_S^{2*}=0$ for all values of $D_P$. Toward obtaining $\alpha^*$, % we re-express $D_C$ and $D_P$ as
%\begin{align}
%D_C=%&\mathbb E\{(X-\mathbb E\{X|Y\})^2\} \nonumber\\
%       %=&\mathbb E\{(X- \frac{1+\alpha \rho}{1+r\alpha^2+2\alpha \rho} (X+\alpha\theta) )^2\} \\
%        =&\sigma_X^2\left (1-\frac{(1+\alpha \rho)^2}{1+r\alpha^2+2\alpha \rho}  \right)
%\end{align}
%
%\begin{align}
%D_P=%&\mathbb E\{(\theta-\mathbb E \{\theta|Y\})^2\} \nonumber \\
%      =&\sigma_X^2\left (r-\frac{(\rho+r\alpha )^2}{1+r\alpha^2+2\alpha \rho}  \right)\nonumber \\
%%            =&\sigma_X^2 \left ( r-\frac{\alpha^2 r^2+\rho^2+2r\alpha \rho  }{1+2\alpha \rho+\alpha^2 r }\right) \\
%%      =&\sigma_X^2 \left ( r-\frac{r(1+2\alpha \rho+\alpha^2 r-1+\rho^2/r)  }{1+2\alpha \rho+\alpha^2 r}\right) \\
%            =&\sigma_X^2 \left ( \frac{r-\rho^2 }{1+2\alpha \rho+\alpha^2 r}\right) \label{eq:24}
%\end{align}
% 
plug  $\sigma_N^2=0$ into (\ref{eq35}), to obtain: 
 $$\left (1-\frac{r-\rho^2}{D_P/\sigma_X^2}\right)+2\alpha \rho+\alpha^2 r=0.$$
 The solution to this second order equation is simply 
 \begin{equation}
 \alpha^*=-\frac{\rho}{r} \pm \sqrt{\left (1-\frac{\rho^2}{r}\right)\left(\frac{1}{r}- \frac{\sigma_X^2}{D_P}\right)}
 \end{equation}
Both of these solutions (corresponding to $\pm$) satisfy the privacy constraint with equality, and the following one achieves lower $D_C$, and hence is the optimal solution:
 \begin{equation}
 \alpha^*=-\frac{\rho}{r} + \sqrt{\left (1-\frac{\rho^2}{r}\right)\left(\frac{1}{r}- \frac{\sigma_X^2}{D_P}\right)}
 \end{equation}
The optimal decoding mapping is 
\begin{align}
h(Y)=\mathbb E\{X|Y\}
=&\frac{1+\alpha \rho}{1+r\alpha^2+2\alpha \rho} Y.
\end{align}
\end{proof}

 \subsection{Problem-2}
Next, we consider compression of jointly  Gaussian source-private information  with entropy based privacy measure and MSE distortion. First, we observe that all equilibrium points are  achieved by a jointly Gaussian $X, \theta, Y$ and hence $Y$ can be written as $Y=X+\alpha \theta+N$ for some $\alpha \in \mathbb R$, and $N \sim\mathbb N(0, \sigma_N^2)$ is Gaussian and independent of $X$ and $\theta$. %Next, we show that $\beta=\alpha$. 
  
%  \begin{lemma}\label{Gopt}
%All equilibrium points of ${\cal RD}_{S}$ are achieved, uniquely, by the jointly Gaussian $Y, X, \theta$ triplet. 
%  \end{lemma}
The proof of this statement follows  from the well-known property of Gaussian distribution achieving maximum entropy under a variance constraint and the steps in the proof of Theorem 1. 
Hence, the test channel achieving the rate-distortion function adds independent Gaussian noise (forward test channel interpretation of Gaussian rate distortion also holds  in this privacy constrained setting). Note that  the privacy constraint is always active in the simple equilibrium setting, and the equilibrium is at the boundary of the constraint set, i.e., we find optimal $\alpha$ by setting the privacy constraint equality. In the compression case, compression itself provides some level privacy inherently (it is evident from the forward channel interpretation of the RD function).  Hence,  the privacy constraint may not be active in the compression case, which yields $\alpha^*=0$. The following theorem characterizes the optimal rate-distortion-privacy trade-off. 
   \begin{theorem}
\label{gaussian}
For a given $D_P$, the space of $R-D_C$ is given as 
\begin{align}
R=&\log \left (1+\frac{\sigma_X^2}{\sigma_N^2} (1+\alpha^2r+2\alpha\rho) \right) \\
%D_P=&\sigma_X^2 \left ( r-\frac{(\rho+r\alpha )^2}{1+2\alpha \rho+\alpha^2 r+\frac{\sigma_N^2}{\sigma_X^2} }\right) \\
D_C=&\sigma_X^2 \left ( 1-\frac{(1+\alpha \rho )^2}{1+2\alpha \rho+\alpha^2 r+\frac{\sigma_N^2}{\sigma_X^2} }\right) 
\end{align}
 where 
 \begin{equation}
 \alpha=-\frac{\rho}{r}+ \Phi,%\sqrt{\left (1-\frac{\rho^2}{r}\right)\left(\frac{1}{r}- \frac{1}{\frac{D_P}{\sigma_X^2} (1+\sigma_N^2)-r\sigma_N^2}\right)}
 \end{equation}
 and 
  \begin{equation}
\Phi= \min \left(\sqrt{\left (1-\frac{\rho^2}{r}\right)\left(\frac{1}{r}- \frac{1}{\frac{D_P}{\sigma_X^2} (1+\sigma_N^2)-r\sigma_N^2}\right)},\frac{\rho}{r} \right) \nonumber
 \end{equation}
 as a function of $\sigma_N^2$.
  \end{theorem}

  \begin{proof}
 We have $Y=X+\beta\theta+N$ for some $\alpha \in \mathbb R$ where $N$ is  zero-mean Gaussian with variance $\sigma_N^2$ and independent of $X$ and $\theta$. This representation yields, by standard estimation theoretic techniques, the following characterization of $R, D_C, D_P$ in terms of $\sigma_N^2$: 
    \begin{align}
R&=\frac{1}{2}\log \left (1+\frac{\sigma_X^2}{\sigma_N^2} (1+\alpha^2 r+2\alpha \rho) \right) \label{eq:1} \\
D_C&=\sigma_X^2 \left ( \frac{   \alpha ^2(r-\rho^2)+\frac{\sigma_N^2}{\sigma_X^2}   }  {1+2\alpha\rho+\alpha^2r+\frac{\sigma_N^2}{\sigma_X^2} }\right )  \label{eq:2}\\
D_P&=\sigma_X^2 \left ( r-\frac{(\rho+r\alpha )^2}{1+2\alpha \rho+\alpha^2 r+\frac{\sigma_N^2}{\sigma_X^2} }\right)  \label{eq:3}
  \end{align}
%Note that $P_T=\mathbb E\{U^2\}=\beta^2\sigma_X^2(1+2\alpha \rho+\alpha^2 r )$, we re-express $D_P$ as 
%\begin{align}
%D_P=&\sigma_X^2 \left ( r-\frac{\beta^2\sigma_X^2(\rho+r\alpha )^2}{P_T+\sigma_N^2 }\right) \\
%       =&\sigma_X^2 \left ( r-\frac{(\rho+r\alpha )^2}{1+2\alpha \rho+\alpha^2 r }\frac{P_T}{P_T+\sigma_N^2 }\right) \\
%      % =&\sigma_X^2 \left ( r-\frac{r (1+2\alpha \rho+\alpha^2 r+\rho^2/r-1 )}{1+2\alpha \rho+\alpha^2 r }\frac{P_T}{P_T+\sigma_N^2 }\right) \nonumber \\
%        =&\sigma_X^2 \left (\frac{\sigma_N^2}{P_T+\sigma_N^2 } r+\frac{r-\rho^2}{1+2\alpha \rho+\alpha^2 r }\frac{P_T}{P_T+\sigma_N^2 }\right) \label{eq:45}
%\end{align}

%
%$$\rho^2-  r\frac{  \frac{P_T (r-\rho^2)}{P_T+\sigma_N^2 } }{ \frac{D_P}{\sigma_X^2}-\frac{\sigma_N^2}{P_T+\sigma_N^2 } r   } $$
%
%$$\rho^2-  \frac{   (r-\rho^2) }{ \frac{D_P}{r\sigma_X^2} (1+\sigma_N'^2)- \sigma_N'^2   } $$

Following steps similar to the ones in the proof of Theorem 2, we can express $\alpha^*$ in terms of $D_P$, when the privacy constraint is active, as: 
 \begin{equation}
 \alpha^*=-\frac{\rho}{r} \pm \sqrt{\left (1-\frac{\rho^2}{r}\right)\left(\frac{1}{r}- \frac{1}{\frac{D_P}{\sigma_X^2} (1+\sigma_N'^2)-r\sigma_N^2}\right)}. \nonumber
 \end{equation}
 Form these solutions, the following achieves a lower value for $D_C$:
  \begin{equation}
 \alpha^*=-\frac{\rho}{r} + \sqrt{\left (1-\frac{\rho^2}{r}\right)\left(\frac{1}{r}- \frac{1}{\frac{D_P}{\sigma_X^2} (1+\sigma_N'^2)-r\sigma_N^2}\right)}.\nonumber
 \end{equation}
 When the privacy constraint is already satisfied, $\alpha^*=0$.

 %The objective of the encoder is to minimize $D_C$ over the choice of $\beta$ which is equivalent to maximizing

  \end{proof}
  
%  \begin{remark}
%   This result parallels the classical rate-distortion result that states R-D optimal reconstruction is Gaussian under quadratic distortion measure \cite{BergerBook}.
%   \end{remark}

 \subsection{Problem-3}
We next focus on noisy communication settings, i.e., we assume there is an additive white Gaussian noise $Z\sim \mathbb N(0, \sigma_Z^2)$ as shown in Figure 2. The following theorem provides the  encoding and decoding mappings at the equilibirum. 
%   \begin{lemma}
%For the quadratic Gaussian communication setting with entropy based privacy constraint, linear encoding strategies are optimal for all power levels. More precisely, $$U= \sqrt {\frac{P_T}{\sigma_X^2(1+2\alpha\rho+\alpha^2r)}}(X+\alpha \theta)$$  for some $\alpha \in \mathbb R$ is the optimal encoding mapping. 
%  \end{lemma}  
%  \begin{proof}

%  \end{proof}

 \begin{theorem}\label{main2}
For the quadratic Gaussian communication setting with entropy based privacy constraint, the (essentially) unique equilibrium is achieved by $$U= \sqrt {\frac{P_T}{\sigma_X^2(1+2\alpha\rho+\alpha^2r)}}(X+\alpha \theta)$$ and $h(Y)=\kappa Y$ where $\alpha$ and $\kappa$ are constants given as: 
\begin{align}
 \kappa=&\frac{1+\alpha\rho}{1+\alpha^2r+2\alpha\rho}
 \end{align}
  and 
 \begin{equation}
 \alpha=-\frac{\rho}{r}+ \Phi,%\sqrt{\left (1-\frac{\rho^2}{r}\right)\left(\frac{1}{r}- \frac{1}{\frac{D_P}{\sigma_X^2} (1+\sigma_N^2)-r\sigma_N^2}\right)}
 \end{equation}
 and 
  \begin{equation}
\Phi= \min \left(\sqrt{\left (1-\frac{\rho^2}{r}\right)\left(\frac{1}{r}- \frac{1}{\frac{D_P}{\sigma_X^2} (1+\sigma_N^2)-r\sigma_N^2}\right)},\frac{\rho}{r} \right) \nonumber
 \end{equation}
 \end{theorem}
\begin{proof}
First, we observe that linear mappings are optimal,  due to the well-known optimality of linear mappings (without the privacy constraint) \cite{goblick}, and the fact that jointly Gaussian  $U,X, \theta$ maximizes  $H(\theta|Y)$ with fixed second-order statistics. Next, we assume without any loss of generality, that $Y=\beta (X+\alpha \theta)+N$ for some $\beta \in \mathbb R^+$ and $\alpha \in \mathbb R$. Then, we have
\begin{align}
\mathbb E\{XY\}=\beta \sigma_X^2 (1+\alpha \rho),&\quad \mathbb E\{\theta Y\}=\beta \sigma_X^2 ( \rho+r\alpha) \nonumber \\
\mathbb E\{Y^2\}&=\beta^2 \sigma_X^2 (1+2\alpha \rho+r\alpha^2)+\sigma_N^2\nonumber
\end{align}
Hence, 
\begin{align}
D_P=&\sigma_X^2 \left ( r-\frac{(\rho+r\alpha )^2}{1+2\alpha \rho+\alpha^2 r+\frac{\sigma_N^2}{\beta^2\sigma_X^2} }\right) \\
D_C=&\sigma_X^2 \left ( 1-\frac{(1+\alpha \rho )^2}{1+2\alpha \rho+\alpha^2 r+\frac{\sigma_N^2}{\beta^2\sigma_X^2} }\right) 
\end{align}
Note that since $P_T=\mathbb E\{U^2\}=\beta^2\sigma_X^2(1+2\alpha \rho+\alpha^2 r )$, we can re-express $D_P$ as 
\begin{align}
D_P=&\sigma_X^2 \left ( r-\frac{\beta^2\sigma_X^2(\rho+r\alpha )^2}{P_T+\sigma_N^2 }\right) \\
       =&\sigma_X^2 \left ( r-\frac{(\rho+r\alpha )^2}{1+2\alpha \rho+\alpha^2 r }\frac{P_T}{P_T+\sigma_N^2 }\right) \\
      % =&\sigma_X^2 \left ( r-\frac{r (1+2\alpha \rho+\alpha^2 r+\rho^2/r-1 )}{1+2\alpha \rho+\alpha^2 r }\frac{P_T}{P_T+\sigma_N^2 }\right) \nonumber \\
        =&\sigma_X^2 \left (\frac{\sigma_N^2}{P_T+\sigma_N^2 } r+\frac{r-\rho^2}{1+2\alpha \rho+\alpha^2 r }\frac{P_T}{P_T+\sigma_N^2 }\right) \label{eq:45}
\end{align}

%\
%$$\rho^2-  r\frac{  \frac{P_T (r-\rho^2)}{P_T+\sigma_N^2 } }{ \frac{D_P}{\sigma_X^2}-\frac{\sigma_N^2}{P_T+\sigma_N^2 } r   } $$
%
%$$\rho^2-  \frac{   (r-\rho^2) }{ \frac{D_P}{r\sigma_X^2} (1+\sigma_N'^2)- \sigma_N'^2   } $$

Note that the power constraint is always active, and when the privacy constraint is active, we have: 
 \begin{equation}
 \alpha^*\!=-\!\frac{\rho}{r} \!\pm \sqrt{\left (\!1\!-\!\frac{\rho^2}{r}\right)\left(\frac{1}{r}\!-\! \frac{1}{\frac{D_P}{\sigma_X^2} (\!1\!+\!\sigma_N^2/P_T)\!-\!r\sigma_N^2/P_T}\right)}. \nonumber
 \end{equation}
 From these solutions, the following achieves a lower  $D_C$:
  \begin{equation}
 \alpha^*\!=\!-\frac{\rho}{r} + \sqrt{\left (1\!-\!\frac{\rho^2}{r}\right)\left(\frac{1}{r}\!-\! \frac{1}{\frac{D_P}{\sigma_X^2} (\!1\!+\!\sigma_N^2/P_T)\!-\!r\sigma_N^2/P_T}\right)}.\nonumber
 \end{equation}
 When the privacy constraint is already satisfied, $\alpha^*=0$.
  \end{proof}

%  \begin{remark}
%In \cite{akyol}, a similar SIT problem was considered, where the transmitter and the receiver have both quadratic but different objectives, and there is no privacy constraint. Although problem settings are similar in the sense that they are both Stackelberg variations of SIT type problems, the equilibrium conditions significantly  differ. In the privacy setting, the coefficient $\alpha$ explicitly depends on the privacy constraint $D_P$, and thereby, the channel power $P_T$ or compression rate $R$ (besides the source statistics), while in \cite{akyol}, it only depends on the source statistics.  
%\end{remark}

%
%  \begin{remark}
%  In this paper, we have studied the quadratic Gaussian case, although for the classical communication problem, it is possible to characterize the conditions (in terms of source and channel distributions, source distortion  and channel cost measures) under which zero-delay communication is optimal, see \cite{gastpar2003code}. A natural follow-up question in the present context is then the following: what are the conditions for privacy, distortion, channel cost constraint and the joint distribution of the source, privacy and the communication channel, so that zero-delay communication is in asymptotic sense optimal?   We leave this question, which is an information theoretic problem in nature, for future work.  
%  \end{remark}
%  

\section{Conclusion}
In this paper, we have addressed some fundamental problems associated with strategic communication in the presence of privacy constraints. Although the compression and communication problems  are inherently information theoretic, for entropy based privacy measure, MSE distortion and jointly Gaussian source and private information, the problem admits a control theoretic representation (optimization over second order statistics). We have explicitly characterized the equilibrium conditions for compression and communication under privacy constraints. Rather surprisingly, the simple equilibrium solution (without compression)  does not require addition of independent noise to satisfy the privacy constraints, as opposed to the common folklore in such problems. Some future directions include using results presented in this paper for decentralized stochastic control problems (see \cite{bansal1987stochastic}) with privacy constraints, extending the approach to vector and network settings, and finally investigating implications in economics. 
  \appendices

\bibliographystyle{IEEEbib}

\bibliography{/Users/eakyol/Dropbox/SIT/ref}

\end{document}